\documentclass[runningheads,a4paper]{llncs}

\usepackage{graphicx}
\usepackage{amssymb}
\usepackage{amsmath}
\usepackage{epsfig}

\newcommand{\eps}{\ensuremath{\epsilon}}

\renewcommand{\setminus}{-}

\begin{document}

\bibliographystyle{plain}
\renewcommand{\thefootnote}{\fnsymbol{footnote}}

\mainmatter

\title{Polylogarithmic Supports are required for Approximate Well-Supported Nash Equilibria below 2/3}

\author{
Yogesh Anbalagan\inst{1}
\and Sergey Norin\inst{2}
\and Rahul Savani\inst{3}
\and Adrian Vetta\inst{4}
}
\institute{
School of Computer Science, McGill University.  {\tt
yogesh.anbalagan@mail.mcgill.ca} \and
Department of Mathematics and Statistics, McGill University.  {\tt
snorin@math.mcgill.ca} \and
Department of Computer Science, University of Liverpool.  {\tt
rahul.savani@liverpool.ac.uk} \and
Department of Mathematics and Statistics, and School of Computer Science, McGill University.\\
{\tt vetta@math.mcgill.ca}
}
\maketitle

\renewcommand{\thefootnote}{\arabic{footnote}}

\begin{abstract}
In an $\epsilon$-approximate Nash equilibrium, a player can gain at most
$\epsilon$ {\em in expectation} by unilateral deviation. An $\epsilon$-{\em
well-supported} approximate Nash equilibrium has the stronger requirement that
every pure strategy used with positive probability must have payoff within
$\epsilon$ of the best response payoff. Daskalakis, Mehta and
Papadimitriou~\cite{DMP09} conjectured that every win-lose bimatrix game has a
$\frac{2}{3}$-well-supported Nash equilibrium that uses supports of cardinality
at most three. Indeed, they showed that such an equilibrium will exist {\em
subject to} the correctness of a graph-theoretic conjecture. Regardless of the
correctness of this conjecture, we show that the barrier of a $\frac23$ payoff
guarantee cannot be broken with constant size supports; we construct win-lose
games that require supports of cardinality at least $\Omega(\sqrt[3]{\log n})$
in any $\epsilon$-well supported equilibrium with $\epsilon < \frac23$. The key
tool in showing the validity of the construction is a proof of a bipartite
digraph variant of the well-known Caccetta-H\"aggkvist conjecture~\cite{CH78}.
A probabilistic argument~\cite{KS10} shows that there exist $\epsilon$-well-supported
equilibria with supports of cardinality $O(\frac{1}{\epsilon^2}\cdot \log n)$,
for any $\epsilon> 0$; thus, the polylogarithmic cardinality bound presented cannot
be greatly improved.
We also show that for any $\delta>0$, there exist
win-lose games for which no pair of strategies with support sizes at most two
is a $(1-\delta)$-well-supported Nash equilibrium. In contrast, every bimatrix
game with payoffs in $[0,1]$ has a $\frac{1}{2}$-approximate Nash equilibrium
where the supports of the players have cardinality at most two~\cite{DMP09}.
\end{abstract}

\section{Introduction}\label{sec:intro}


\renewcommand{\thefootnote}{\arabic{footnote}}

\renewcommand{\theenumi}{(\alph{enumi})}
\renewcommand{\labelenumi}{(\alph{enumi})}


A Nash equilibrium of a bimatrix game is a pair of strategies in which the supports of both
players consist only of best responses.
The apparent hardness of computing an exact Nash equilibrium~\cite{DGP09,CDT09}
even in a bimatrix game has led to work on computing approximate Nash
equilibria, and two notions of approximate Nash equilibria have been developed.
The first and more widely studied notion
is of an \emph{$\epsilon$-approximate Nash equilibrium} ($\epsilon$-Nash).
Here, no restriction is placed upon the supports; any strategy can be in the supports
provided each player achieves an expected payoff that is within~$\epsilon$
of a best response. Therefore, $\epsilon$-Nash equilibria have a practical drawback: a player
might place probability on a strategy that is arbitrarily far from being a best response.
The second notion, defined to rectify this problem,  is called an
\emph{$\epsilon$-well supported approximate Nash equilibrium} ($\epsilon$-WSNE).
Here, the content of the supports are restricted, but less stringently than in
an exact Nash equilibrium.
Specifically, both players can only place positive probability on strategies that have payoff
within $\epsilon$ of a pure best response.
Observe that the latter notion is a stronger equilibrium concept: every $\epsilon$-WSNE is an $\epsilon$-Nash,
but the converse is not true.

Approximate well-supported equilibria recently played an important role in understanding the
hardness of computing Nash equilibria. They are more useful in these contexts
than $\eps$-Nash equilibria because their definition is more combinatorial and more closely resembles
the {\em best response condition} that characterizes exact Nash equilibria.
Indeed, approximate well-supported equilibria were introduced in~\cite{GP06,DGP09} in
the context of PPAD reductions that show the hardness of computing (approximate) Nash equilibria.
They were subsequently used as the notion of approximate equilibrium by Chen et
al.~\cite{CDT09} that showed the PPAD-hardness of computing an exact Nash
equilibrium even for bimatrix games.

Another active area of research is to investigate the best (smallest) $\epsilon$ that can be
guaranteed in polynomial time.
For $\epsilon$-Nash, the current best algorithm, due to Tsaknakis and Spirakis~\cite{TS08}, achieves a
$0.3393$-Nash equilibrium; see \cite{DMP09,DMP07,BBM10} for other algorithms.
For the important class of win-lose games -- games with payoffs in $\{0, 1\}$ -- \cite{TS08} gives
a $\frac{1}{4}$-Nash equilibrium.
For \eps-WSNE, the current best result was given by Fearnley et al.~\cite{FGS12} and finds a
$(\frac{2}{3}-\zeta)$-WSNE, where $\zeta = 0.00473$.
It builds on an approach of Kontogiannis and Spirakis~\cite{KS10},
which finds a $\frac{2}{3}$-WSNE in polynomial time using linear programming.
The algorithm of Kontogiannis and Spirakis produces a $\frac{1}{2}$-WSNE of
win-lose games in polynomial time, which is best-known
(the
modifications of Fearnley et al.\  do not lead to an improved approximation
guarantee for win-lose games).

It is known that this line of work cannot extend to a fully-polynomial-time
approximation scheme (FPTAS).
More precisely, there does not exist an FPTAS for computing approximate Nash equilibria unless
PPAD is in P~\cite{CDT09}.
Recall, an FPTAS requires a running time that is polynomial in both the size of
the game input and in $\frac{1}{\eps}$. A polynomial-time approximation scheme (PTAS), however, need
not run in time polynomial in $\frac{1}{\eps}$. It is not known whether
there exists a PTAS for computing an approximate Nash equilibrium and,
arguably, this is the biggest open question in equilibrium computation today.
While the best-known approximation guarantee for \eps-Nash that is achievable
in polynomial time is much better than that
for \eps-WSNE, the two notions are polynomially related:
there is a PTAS for $\epsilon$-WSNE if and only if there is a
PTAS for $\epsilon$-Nash~\cite{CDT09,DGP09}.

\subsection{Our Results}
The focus of this paper is on the combinatorial structure of equilibrium.
Our first result shows that well-supported Nash equilibria differ
structurally from approximate Nash equilibria in a significant way.
It is known that there are $\frac12$-Nash equilibria with supports
of cardinality at most two \cite{GP06}, and that this result is tight
\cite{FNS07}.
In contrast, we show in Theorem~\ref{thm:two} that for any $\delta>0$, there exist win-lose games for
which no pair of strategies with support sizes at most two is a
$(1-\delta)$-well-supported Nash equilibrium.\footnote{Random games have been shown to have exact equilibria with
support size $2$ with high probability; see B\'{a}r\'{a}ny et al.~\cite{BVV07}.}

With supports of cardinality three, Daskalakis et
al. conjectured, in the first paper that studied algorithms for finding $\eps$-WSNE~\cite{DMP09},
that $\frac{2}{3}$-WSNE are obtainable in every win-lose bimatrix game.
Specifically, this would be a consequence of the following graph-theoretic conjecture.
\begin{conjecture}[\cite{DMP09}]
\label{con}
Every digraph either has a cycle of length at most 3 or an undominated set\footnote{A set $S$ is undominated if there is no vertex $v$
that has an arc to every vertex in $S$.} of 3 vertices.

\end{conjecture}

The main result in this paper, Theorem~\ref{thm:main}, shows that one cannot do better with constant size
supports.
We prove that there exist win-lose games that require supports of cardinality at
least $\Omega(\sqrt[3]{\log n})$ in any $\epsilon$-WSNE with $\epsilon < \frac23$.
We prove this existence result probabilistically. The key tool in showing correctness is
a proof of a bipartite digraph variant of the well-known Caccetta-H\"aggkvist conjecture~\cite{CH78}.

A polylogarithmic cardinality bound, as presented here, is the best we can hope for -- a probabilistic argument \cite{KS10} shows that
there exist $\epsilon$-WSNE with
supports of cardinality $O(\frac{1}{\epsilon^2}\cdot \log n)$, for any
$\epsilon> 0$.\footnote{For \eps-Nash equilibria, Alth\"ofer~\cite{Alt94} and
Lipton and Young \cite{LY94} independently proved similar results for zero-sum
games, and Lipton at al.~\cite{LMM03} later proved similar results for 
general-sum bimatrix and multi-player games.}

\section{A Lower Bound on the Support Size of Well Supported Nash Equilibria}

We begin by formally defining bimatrix win-lose games and well-supported Nash equilibria.
A {\em bimatrix game} is a $2$-player game with $m\times n$ payoff matrices $A$ and $B$; we may assume that $m \le n$.
The game is called {\em win-lose} if each matrix entry is in $\{0, 1\}$.

A pair of mixed strategies $\{{\bf p}, {\bf q}\}$ is a {\em Nash equilibrium} if
every pure row strategy in the support of {\bf p} is a best response to
{\bf q} and every pure column strategy in the support of {\bf q} is a best response to
{\bf p}. A relaxation of this concept is the following.
A pair of mixed strategies $\{{\bf p}, {\bf q}\}$ is an {\em $\epsilon$-well supported Nash equilibrium} if
every pure strategy in the support of {\bf p} (resp.  {\bf q}) is an $\epsilon$-approximate best response to
{\bf q} (resp.  {\bf p}). That is, for any row $r_i$ in the support of ${\bf p}$ we have
$${\bf e_i}^T A{\bf q} \ge \max_{\ell} {\bf e_{\ell}}^T A{\bf q} -\epsilon$$
and, for any column $c_j$ in the support of ${\bf q}$ we have
$${\bf p}^T B{\bf e_j} \ge \max_{\ell} {\bf p}^T B{\bf e_{\ell}} -\epsilon\ .$$

In this section we prove our main result.

\begin{theorem}\label{thm:main}
 For any $\epsilon <\frac23$, there exist win-lose games for which every
$\epsilon$-well-supported Nash equilibrium has supports of cardinality $\Omega(\sqrt[3]{\log n})$.
\end{theorem}

To prove this result, we first formulate our win-lose games graphically.
This can be done in a straight-forward manner. Simply observe that we may represent a $2$-player win-lose game
by a directed bipartite graph $G=(R\cup C, E)$.
There is a vertex for each row and a vertex for each column. There is an arc $(r_i,c_j)\in E$
if and only if $(B)_{ij}=b_{ij}=1$; similarly there is an arc $(c_j,r_i)\in E$ if and only if $(A)_{ij}=a_{ij}=1$.

Consequently, we are searching for a graph whose corresponding game has no high quality
well-supported Nash equilibrium with small supports. We show the existence of such a graph
probabilistically. \\

\noindent{\tt The Construction.}\\
Let $T=(V,E)$ be a random tournament on $N$ nodes. Now create from $T$ an auxiliary bipartite graph $G(T)=(R\cup C, A)$ corresponding
to a $2$-player win-lose game as follows.
The auxiliary graph has a vertex-bipartition $R\cup C$ where there is a vertex of $R$ for each node of $T$ and there is a vertex of $C$ for each
set of $k$ distinct nodes of $T$. (Observe that, for clarity we will refer to nodes in the tournament $T$ and vertices in the bipartite graph $G$.)
There are two types of arc in $G(T)$: those oriented from $R$ to $C$ and those oriented from $C$ to $R$.
For arcs of the former type, each vertex $X\in C$ will have in-degree exactly $k$.
Specifically, let $X$ correspond to the $k$-tuple $\{v_1,\dots, v_k\}$ where $v_i\in V(T)$, for all $1\le i\le k$. Then there are arcs
$(v_i, X)$ in $G$ for all $1\le i\le k$.
Next consider the latter type of arc in $G$.
For each node $u\in R$ there is an arc $(X, u)$ in $G$ if and only if
$u$ dominates $X=\{v_1,\dots, v_k\}$ in the tournament $T$, that is if $(u,v_i)$ are arcs in $T$ for all $1\le i\le k$.
This completes the construction of the auxiliary graph (game) $G$.

We say that a set of vertices $W=\{w_1,\dots, w_t\}$ is {\em covered} if there exists a vertex $y$ such that $(w_j, y)\in A$, for all $1\le j\le t$.
Furthermore, a bipartite graph is {\em $k$-covered} if every collection of $k$ vertices that lie on the same side of the bipartition is covered.
Now with positive probability the auxiliary graph $G(T)$ is $k$-covered.

\begin{lemma}\label{lem:randomT}
For all  sufficiently large $n$ and $k\le \sqrt[3]{\log n}$, there exists a tournament $T$ whose auxiliary bipartite graph $G(T)$ is $k$-covered.
\end{lemma}
\begin{proof}
Observe that the payoff matrices that correspond to $G(T)$ have $m=N$ rows and $n={N \choose k}$ columns.
Furthermore, by construction, any set of $k$ vertices in $R$ is covered.
Thus, first we must verify that any set of $k$ vertices in $C$ is also covered.

So consider a collection $\mathcal{X}=\{X_1,\dots, X_k\}$ of $k$ vertices in $C$.
Since each $X_i\in C$ corresponds to a $k$-tuple of nodes of $T$, we see that $\mathcal{X}$ corresponds to
a collection of at most $k^2$ nodes in $T$. Thus, for any node $u\notin \cup_i X_i$,
we have that $u$ has an arc in $T$ to every node in $\cup_i X_i$ with probability at least $2^{-k^2}$.
Thus with probability at most $(1-\frac{1}{2^{k^2}})^{N-k^2}$ the subset $\mathcal{X}$ of $C$ not covered in $G(T)$.
Applying the union bound we have that there exists the desired tournament if
\begin{equation}\label{eq:prob}
{ n \choose k}\cdot \left(1-\frac{1}{2^{k^2}}\right)^{N-k^2} < 1
\end{equation}
Now set $k={\log^{\frac 13} n}$. Therefore $\log n^{\frac{1}{k}} = \log^{\frac23} n = k^2$.

In addition, because $n={N \choose k}$, we have that $N\ge \frac{k}{e}\cdot n^{\frac{1}{k}}$.
Hence, $N-k^2> n^{\frac{1}{k}}$. (Note that, since $N\ge k$ this implies that
$G(T)$ is defined.) Consequently,
\begin{eqnarray*}
 { n \choose k}\cdot \left(1-\frac{1}{2^{k^2}}\right)^{N-k^2}
&\le& n^k\cdot \left(1-\frac{1}{2^{k^2}}\right)^{n^{\frac{1}{k}}} \\
&\le& n^k\cdot e^{-\frac{1}{2^{k^2}}\cdot {n^{\frac{1}{k}}}}\\
&\le& n^k\cdot e^{-\frac{1}{e^{k^2\cdot \log 2}}\cdot {n^{\frac{1}{k}}}}
\end{eqnarray*}
Thus, taking logarithms, we see that Inequality (\ref{eq:prob}) holds if
\begin{equation}\label{eq:prob2}
e^{k^2\cdot \log 2} \cdot k\cdot \log n < n^{\frac{1}{k}}
\end{equation}
But $n^{\frac{1}{k}}=e^{k^2}$, so Inequality (\ref{eq:prob2}) clearly holds
for large $n$.
The result follows.
\end{proof}

A property of the  auxiliary graph $G(T)$ that will be very useful to us is that it contains no cycles with less than six vertices.
\begin{lemma}\label{lem:no4cycle}
The auxiliary graph $G(T)$ contains no digons and no $4$-cycles.
\end{lemma}
\begin{proof}
Suppose $G(T)$ contains a digon $\{w, X\}$. The arc $(w, X)$ implies that $X=\{x_1,\dots, x_{k-1}, w\}$.
On the other-hand, the arc $(X,w)$ implies that $w$ dominates $X$ in $T$ and, thus, $w\notin X$.

Suppose $G(T)$ contains a $4$-cycle $\{w, X, z, Y\}$ where $w$ and $z$ are in $R$ and where
$X=\{x_1,\dots, x_{k-1}, w\}$ and $Y=\{y_1,\dots, y_{k-1}, z\}$ are in $C$.
Then $z$ must dominate $X$ in $T$ and  $w$ must dominate $Y$ in $T$.
But then we have a digon in $T$ as $(w,z)$ and $(z,w)$ must be arcs in $T$. This contradicts
the fact that $T$ is a tournament.
\end{proof}

Lemmas \ref{lem:randomT} and \ref{lem:no4cycle} are already sufficient to prove a major distinction between
approximate-Nash equilibria and well-supported Nash equilibria. Recall that there always exist
$\frac12$-Nash equilibria with supports of cardinality at most two \cite{DMP09}.
In sharp contrast, for supports of cardinality at most two, no constant approximation
guarantee can be obtained for \eps-well-supported Nash equilibria.
\begin{theorem}\label{thm:two}
For any $\delta>0$, there exist win-lose games for which no pair of strategy vectors
with support sizes at most two is a $(1-\delta)$-well-supported Nash equilibrium.
\end{theorem}
\begin{proof}
Take the auxiliary win-lose game $G(T)$ from Lemma \ref{lem:randomT} for the case $k=2$.
Now consider any pair of strategy vectors ${\bf p_1}$ and ${\bf p_2}$ with supports of cardinality $2$ or less.
Since $G(T)$ is $2$-covered, the best responses to ${\bf p_1}$ and ${\bf p_2}$ both generate payoffs of
exactly $1$. Thus $\{{\bf p_1}, {\bf p_2}\}$ can be a $(1-\delta)$-well-supported Nash equilibrium
only if each strategy in the support of ${\bf p_1}$ is a best response to at least one of the pure strategies
in the support of  ${\bf p_2}$ and vice versa. Therefore, in the subgraph $H$ of $G(T)$ induced by the supports
of ${\bf p_1}$ and ${\bf p_2}$, each vertex has in-degree at least one. Thus, $H$ contains a directed cycle.
But $G(T)$ has no digons or $4$-cycles, by Lemma \ref{lem:no4cycle}.
Hence, we obtain a contradiction as $H$ contains at most four vertices.
\end{proof}

In light of Lemma \ref{lem:no4cycle}, we will be interested in the minimum in-degree required to ensure that a
bipartite graph contains a $4$-cycle. The following theorem may be of interest on its own right, as it resolves a variant of
the well-known Caccetta-H\"aggkvist conjecture~\cite{CH78} for bipartite digraphs.
For Eulerian graphs, a related but different result is due to Shen and Yuster~\cite{SY02}.

\begin{theorem}\label{thm:extremal}
Let $H=(L\cup R, A)$ be a directed $k\times k$ bipartite graph. If $H$ has minimum in-degree $\lambda \cdot k$ then
it contains a $4$-cycle, whenever $\lambda > \frac13$.
\end{theorem}
\begin{proof}
To begin, by removing arcs we may assume that every vertex has in-degree exactly $\lambda \cdot k$.
Now take a vertex $v$ with the maximum out-degree in $H$, where without loss of generality $v\in L$.
Let $A_1$ be the set of out-neighbours of $v$, and set $\alpha_1\cdot k=|A_1|$.
Similarly, let $B_t$ be the set of vertices with paths to $v$ that contain exactly $t$ arcs, for $t\in \{1, 2\}$, and set
set $\beta_t\cdot k=|B_t|$.
Finally, let $C_1$ be the vertices in $R$ that are not adjacent to $v$, namely $C_1=L\setminus (A_1\cup B_1)$.
Set $\gamma_1\cdot k =|C_1|$.

These definitions are illustrated in Figure \ref{fig:picture}.

\ \\
\vspace{-1cm}
\begin{figure}[h]
\unitlength1cm
\begin{center}
\begin{picture}(10,8)
\epsfig{file=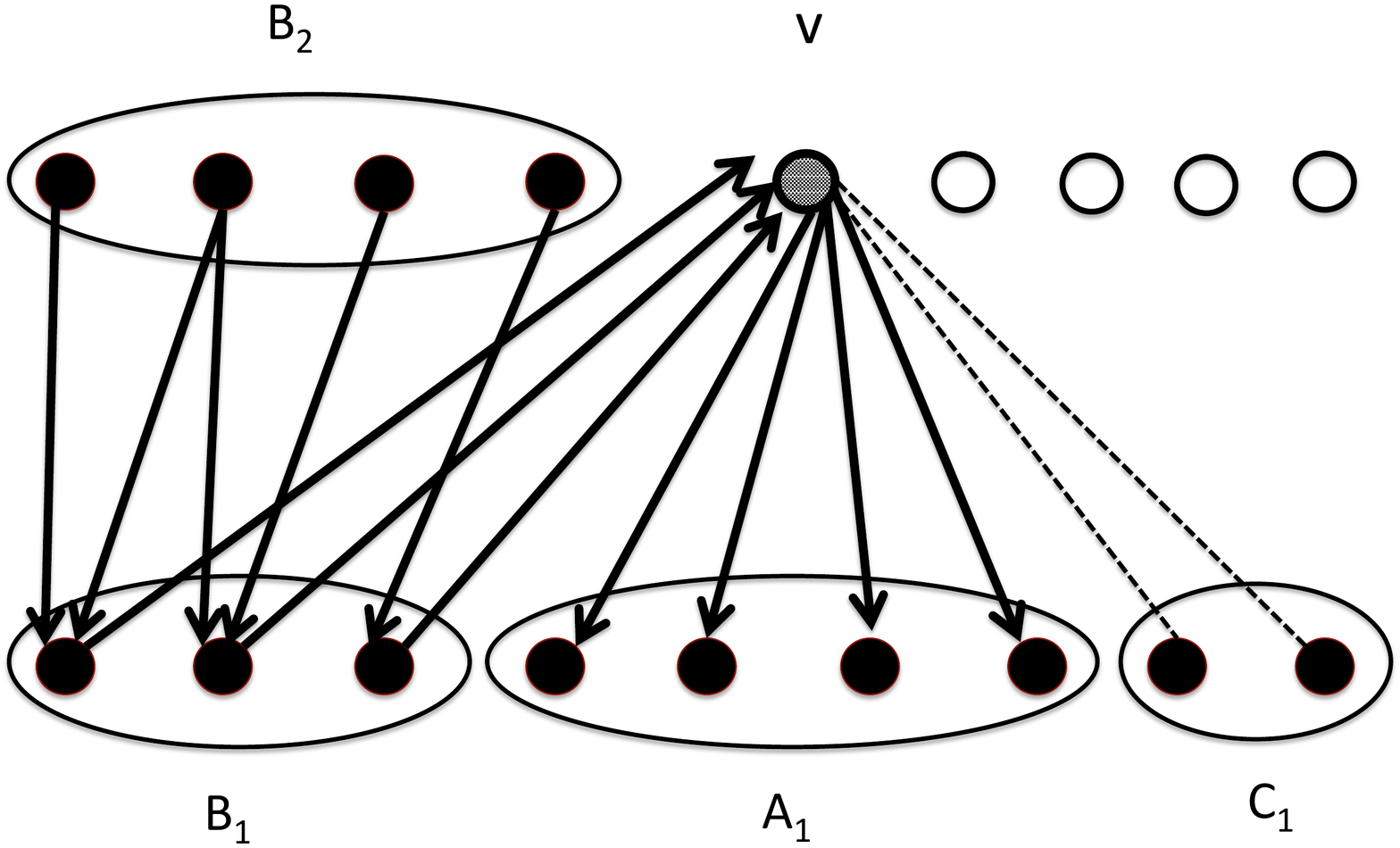, height=8cm, width=10cm}
\end{picture}\par
\vspace{-1cm}
\caption{\label{fig:picture}}
\end{center}
\end{figure}

Observe that we have the following constraints on $\alpha_1, \beta_1$ and $\gamma_1$.
By assumption, $\beta_1= \lambda $. Thus, we have $\gamma_1= 1-\alpha_1-\lambda$.
Moreover, by the choice of $v$, we have $\alpha_1\ge \lambda $, since the maximum out-degree must
be at least the average in-degree.

Note that if there is an arc from $A_1$ to $B_2$ then $H$ contains a $4$-cycle.
So, let's examine the in-neighbours of $B_2$.
We know $B_2$ has exactly $\lambda \cdot k\cdot |B_2|$ incoming arcs.
We may assume all these arcs emanate from $B_1\cup C_1$.
On the other-hand, there are exactly $\lambda \cdot k\cdot |B_1|$ arcs
from $B_2$ to $B_1$. Thus, there are at most $|B_1|\cdot \left( |B_2|-\lambda \cdot k\right)$
arcs from $B_1$ to $B_2$.
So the number of arcs from $C_1$ to $B_2$ is at least
\begin{eqnarray*}
\lambda \cdot k\cdot |B_2| - |B_1|\cdot \left( |B_2|-\lambda \cdot k\right) &=& \lambda \cdot k\cdot \beta_2\cdot k
- \beta_1\cdot k \cdot \left(\beta_2\cdot k +\lambda \cdot k\right)\\
&=& \lambda \cdot k\cdot \beta_2\cdot k
- \lambda\cdot k \cdot \left(\beta_2\cdot k +\lambda \cdot k\right)\\
&=&\lambda^2\cdot k^2
\end{eqnarray*}
Since the maximum degree is $\alpha_1\cdot k$, the number of arcs emanating from
$C_1$ is at most $\gamma_1 \cdot \alpha_1 \cdot k^2$.
Thus $\gamma_1 \cdot \alpha_1 \cdot (1-\alpha_1-\lambda) \ge \lambda^2$.
 Rearranging we obtain the quadratic inequality
\begin{equation*}
\alpha_1^2-\alpha_1(1-\lambda)+\lambda^2 \le 0
\end{equation*}
The discriminant of this quadratic is $1-2\lambda-3\lambda^2$.
But $1-2\lambda-3\lambda^2=(1-3\lambda)(1+\lambda)$ and this
is non-negative if and only if $\lambda\le \frac13$.
This completes the proof.
\end{proof}

We may now prove our main result: no approximation guarantee better than $\frac23$ can be achieved unless the
well-supported equilibria has supports with cardinality $\Omega(\sqrt[3]{\log n})$. \\

\noindent{\bf Proof of Theorem \ref{thm:main}.}
Take a tournament $T$ whose auxiliary bipartite graph is $k$-covered.
By Lemma \ref{lem:randomT}, such a tournament exists.
Consider the win-lose game corresponding to the auxiliary graph $G(T)$,
and take strategy vectors ${\bf p_1}$ and ${\bf p_2}$ with supports of cardinality $k$ or less.
Without loss of generality, we may assume that $\bf{p_1}$ and $\bf{p_2}$ are rational.
Denote these supports as $S_1\subseteq R$ and $S_2\subseteq C$, respectively.
As $G(T)$ is $k$-covered, there is a pure strategy $c^*\in C$ that covers $S_1$
and a pure strategy $r^*\in R$ that covers $S_2$.
Thus, in the win-lose game, $c^*\in C$ has an expected payoff of $1$ against ${\bf p_1}$ and
$r^*\in R$ has an expected payoff of $1$ against ${\bf p_2}$.

Suppose ${\bf p_1}$ and ${\bf p_2}$ form a $\epsilon$-well-supported equilibrium for some $\epsilon<\frac23$.
Then it must be the case that each $r_i\in S_1$ has expected payoff at least $1-\epsilon >\frac13$ against ${\bf p_2}$.
Similarly, each $c_j\in S_2$ has expected payoff at least $1-\epsilon >\frac13$ against ${\bf p_1}$.
But this cannot happen. Consider the subgraph of $G(T)$ induced by $S_1\cup S_2$ where each $r_i\in S_1$ has weight
$w_i=p_1(r_i)$ and each $c_j\in S_2$ has weight $w_j=p_2(c_j)$. We convert this into an unweighted graph $H$ by making
$L\cdot w_v$ copies of each vertex $v$, for some large integer $L$.
Now $H$ is an $L\times L$ bipartite graph with minimum in-degree $(1-\epsilon)\cdot L >\frac13\cdot L$.
Thus, by Theorem \ref{thm:extremal}, $H$ contains a $4$-cycle. This is a contradiction, by Lemma~\ref{lem:no4cycle}.
\qed

We remark that the $\frac23$ in Theorem \ref{thm:main} cannot be improved using this proof technique.
Specifically the minimum in-degree requirement of $\frac13\cdot k$ in Theorem \ref{thm:extremal} is tight.
To see this, take a directed $6$-cycle $C$ and replace
each vertex in $C$ by $\frac13\cdot k$ copies. Thus each arc in $C$ now corresponds to
a complete $\frac{k}{3}\times \frac{k}{3}$ bipartite graph with all arc orientations in the same direction.
The graph $H$ created in this fashion is bipartite with all in-degrees (and all out-degrees) equal to $\frac13\cdot k$.
Clearly the minimum length of a directed cycle in $H$ is six.

\section{Conclusion}
An outstanding open problem is whether any constant approximation guarantee
better than 1 is
achievable with constant cardinality supports. We have shown that supports of
cardinality two cannot achieve this; a positive resolution of Conjecture~\ref{con}
would suffice to show that supports of cardinality three can.
However, Conjecture~\ref{con} seems a hard graph problem and it is certainly conceivable that it is
false.\footnote{For example, the conjecture resembles a question about the existence of
$k$-existentially complete triangle-free graphs for $k>3$ referred to in \cite{EL13},
which the authors consider to be wide open.} If so, that would lead to the intriguing possibility of a very major
structural difference between $\epsilon$-Nash and $\epsilon$-WSNE; namely, that
for any $\delta>0$, there exist win-lose games for
which no pair of strategies with constant cardinality supports  is a
$(1-\delta)$-well-supported Nash equilibrium.

The existence of small support $\epsilon$-WSNE clearly implies the existence of
of polynomial time approximation algorithms to find such equilibria.
Obtaining better approximation guarantees using more complex algorithms
is also an interesting question.
As discussed, the best known polynomial-time approximation algorithm for well-supported
equilibria in win-lose games finds a $\frac{1}{2}$-well supported
equilibrium~\cite{KS10} by solving a linear program (LP).
For games with payoffs in $[0,1]$ that algorithm finds a
$\frac{2}{3}$-well-supported equilibrium.
The algorithm has been modified in \cite{FGS12} to achieve a slightly better
approximation of about $\frac{2}{3}-\zeta$ where
$\zeta = 0.00473$.
That modification solves an almost identical LP as~\cite{KS10} and then
either transfers probability mass within the supports of a solution to the LP
or returns a small support strategy profile that uses at most two pure
strategies for each player. The results of this paper show that both parts of
that approach are needed, and any improvement to the approximation guarantee
must allow for super-constant support sizes.

\ \\
\noindent{\bf Acknowledgements.}
We thank John Fearnley and Troels S\o rensen for useful discussions. The first author is supported by a fellowship from MITACS and an NSERC grant. The second author is supported by an NSERC grant 418520. The third author is supported by an EPSRC grant EP/L011018/1. The fourth author is supported by NSERC grants 288334 and 429598.

\end{document}